%% file: main.tex
\documentclass[twoside,leqno,twocolumn]{article}

\usepackage[letterpaper]{geometry}

\usepackage{ltexpprt}
\usepackage{hyperref}
\usepackage{common}
\usepackage{mymacro}
\usepackage[hide]{notes}
\usepackage{flushend}

\newtheorem{problem}{Problem}[section]
\newtheorem{remark}{Remark}[section]

\newif\ifsupp 
\supptrue 
\ifsupp\else  \fi 

\begin{document}

\newcommand\relatedversion{}


\title{{\Large Ranking with submodular functions on the fly}\thanks{
This research is supported by the Academy of Finland projects MALSOME (343045), AIDA (317085) and MLDB (325117),
the ERC Advanced Grant REBOUND (834862), 
the EC H2020 RIA project SoBigData++ (871042), 
and the Wallenberg AI, Autonomous Systems and Software Program (WASP) 
funded by the Knut and Alice Wallenberg Foundation.
}}
\author{Guangyi Zhang\thanks{KTH Royal Institute of Technology. \href{mailto:guaz@kth.se}{guaz@kth.se}}
\and Nikolaj Tatti\thanks{HIIT, University of Helsinki. \href{mailto:nikolaj.tatti@helsinki.fi}{nikolaj.tatti@helsinki.fi}}
\and Aristides Gionis\thanks{KTH Royal Institute of Technology. \href{mailto:argioni@kth.se}{argioni@kth.se}}
}

\date{}

\maketitle


\fancyfoot[R]{\scriptsize{Copyright \textcopyright\ 2023 by SIAM\\
Unauthorized reproduction of this article is prohibited}}





\begin{abstract} \small\baselineskip=9pt 
Maximizing submodular functions have been studied extensively 
for a wide range of subset-selection problems.
However, much less attention has been given to the role of submodularity
in sequence-selection and ranking problems.
A recently-introduced framework, named \emph{maximum submodular ranking} (MSR), 
tackles a family of ranking problems that arise naturally
when resources are shared among multiple demands with different budgets.
For example, the MSR framework can be used to rank web pages for multiple user intents.
In this paper, we extend the MSR framework in the streaming setting. 
In particular, we consider two different streaming models
and we propose practical approximation algorithms.
In the first streaming model, called \emph{function arriving},
we assume that submodular functions (demands) arrive continuously in a stream, 
while in the second model, called \emph{item arriving}, 
we assume that items (resources) arrive continuously.
Furthermore, we study the MSR problem with additional constraints on the output sequence, 
such as a matroid constraint that can ensure fair exposure among items from different groups.
These extensions significantly broaden the range of problems that can be captured by the MSR framework.
On the practical side, we develop several novel applications based on the MSR formulation, 
and empirically evaluate the performance of the proposed~methods.
\end{abstract}

\section{Introduction}
\label{section:intro}
\input{intro}

\section{Related work}
\label{section:related}

\input{related}

\section{Problem definition}
\label{section:definition}
\input{definition}

\section{Function-arriving \msr}
\label{section:msrf}
\input{function-arriving}

\section{Item-arriving \msr}
\label{section:msri}
\input{reduction}


\section{Experiments}
\label{section:experiment}
\input{experiment}

\section{Conclusions}
\label{section:conclusion}
\input{conclusion}

\bibliographystyle{abbrvnat}
\bibliography{references}

\ifsupp
\clearpage
\appendix
\section{Appendix}
\input{supp}
\fi

\end{document}

%% file: intro.tex
Submodular set functions capture a ``diminishing-returns'' property that 
is present in many real-world phenomena~\citep{krause2014submodular}.
Submodular functions are popular
as they admit a rich toolbox of optimization techniques developed in the literature.
Examples of submodular functions used in practical problems include 
document summarization \citep{lin2011class}, 
viral marketing in social networks \citep{kempe2015maximizing}, 
social welfare maximization \citep{vondrak2008optimal}, and many other.
The majority of existing submodularity-based problem formulations are restricted 
to selecting a \emph{subset of items}, and completely disregard the effect of \emph{item order}.
However, the order of items plays an important role in many applications.
In this paper, we investigate a versatile approach to ranking items within the submodularity framework.

Creating a sequence of resources to be shared among multiple demands 
appears in a broad range of applications. 
For example, when ranking web pages in response to a user query we want to cater for 
multiple user intents; 
when creating a live stream of music content shared among a group 
we want to satisfy the tastes of all listeners; 
and when selecting advertising for a screen on public display 
we want it to be relevant for all passengers. 
Typically, each demand has an individual maximum budget, 
e.g., in the previous scenario, the budget models the number of web pages that a user 
is expected to browse. 
The main challenge in these problems is to find a (partial) ranking of resources that best satisfies multiple demands with different budgets.

More concretely, the \emph{maximum submodular ranking} (\msr) formulation~\citep{zhang2022ranking}
deals with \emph{resources} and \emph{demands}.
A resource is referred to as an \emph{item} in a universe set \V.
Demands require resources, and the \emph{utility} of a demand for a set of resources 
is characterized by a \emph{non-decreasing submodular set function} $f: 2^\V \to \reals_+$.
The \emph{budget} of a demand is represented by a cardinality constraint, 
i.e., the maximum number of items its corresponding function is allowed to take.
The objective is to find a (partial) sequence of items to maximize the total utility, 
i.e., the sum of function values.
A formal problem definition is introduced in Section~\ref{section:definition}.

To capture a wider range of problems and increase the versatility of the framework, 
in this paper, we extend the \msr problem in two natural streaming models, 
which we name \emph{function arriving} and \emph{item arriving}.
In the \emph{function-arriving model}, 
we assume that demands arrive continuously in an online fashion 
and one is unaware of the type and/or volume of future demands.
This model is common in practice; for example, new audience may join a live stream in any time.
In the \emph{item-arriving model}, we assume that items arrive in a stream, 
and we have to output a sequence of items in \emph{one pass} and 
with limited memory after seeing all the items in the stream.
In other words, we need to process each item immediately after its arrival.
This model offers a way to handle the \msr problem when items are arriving continuously,
or are too many to be loaded into memory.

For both of the streaming models we consider we propose practical approximation algorithms.
Besides, we study the \msr problem with additional constraints on the output sequence, 
such as a matroid constraint (see Section~\ref{section:definition})  
that can ensure fair exposure among items from different groups.

On the practical side, we propose novel applications based on the \msr formulation.
We highlight here an application on \emph{progressively-diverse} personalized recommendation, 
while many other interesting ones are discussed in Section~\ref{section:experiment},
including live streaming and catalogued viral marketing.
A popular approach to personalized recommendation \citep{mitrovic2017streaming} 
is to select a succinct subset $S$ of items that maximizes a weighted sum 
(with a trade-off parameter\,\Ctrade) of two submodular terms,
relevance and diversity.
Note that for a given value of the parameter \Ctrade,
a subset $S$ presents a fixed trade-off between relevance and diversity. 
On the other hand, a user's need for diversity may be better served in an adaptive manner. 
For example, a target user may appreciate a recommended list having the 
\emph{most relevant items at the beginning} and becoming \emph{progressively diverse down the list}.
This requirement can be satisfied via an \msr formulation, 
by maximizing a weighted sum of multiple such functions, each with an increasing trade-off value \Ctrade.
See Section~\ref{section:definition} for a formal definition.

Our contributions in this paper are summarized as follows.
\begin{itemize}
	\item We study the {maximum submodular ranking} (\msr) problem
	in two different streaming models, 
	and devise approximation algorithms for each model.
	\item In the function-arriving streaming model, we show that a simple greedy algorithm yields 
	a 2 approximation for the \msr problem, if an item is allowed to be used multiple times.
	This approximation ratio is tight for the greedy algorithm.
	We also show that the problem is inapproximable if every item can be used at most once.
	\item We propose a novel reduction that maps the ranking problem into a constrained
	subset-selection problem subject to a bipartite-matching constraint.
	An immediate consequence is that there exist approximation algorithms for the \msr problem 
	subject to a general \nm-matroid constraint on the output sequence.
	Another consequence is that we can obtain efficient approximation algorithms 
	for the \msr problem in the item-arriving streaming~model.
	\item We apply the enhanced \msr framework to several novel real-life applications, and
	empirically evaluate the performance of the proposed algorithms.
\end{itemize}

The rest of the paper is organized as follows.
We discuss related work in Section~\ref{section:related}.
A formal problem definition is introduced in Section~\ref{section:definition}.
We present approximation algorithms for the function-arriving \msr problem and 
the item-arriving \msr problem in Sections~\ref{section:msrf} and~\ref{section:msri}, respectively.
Our empirical evaluation is conducted in Section~\ref{section:experiment}, 
followed by concluding remarks in Section~\ref{section:conclusion}.
Missing proofs and further experimental details are deferred to the supplementary materials.
Our implementation is made publicly available.\code

%% file: related.tex
\para{Submodularity for sequences.}
The \msr problem was proposed by \citet{zhang2022ranking}, and 
it was later shown to be a special case of \emph{ordered submodularity}, 
introduced by \citet{kleinberg2022ordered}.
In particular, both formulations avoid an unnatural \emph{postfix monotonicity} property, 
which is required in prior formulations \citep{streeter2008online,zhang2012submodularity,alaei2021maximizing}.
Postfix monotonicity requires a non-decreasing function value after prepending an arbitrary item at the front of a sequence.
Additionally, the \msr problem can be seen as a dual problem to the \emph{submodular ranking} problem \citep{azar2011ranking}, which aims to minimize the total cover time of all functions in the absence of individual budgets.
No streaming extension has been known for the \msr problem.

\para{Bipartite matching.}
Bipartite matching and its many variants have been extensively studied 
in the literature~\citep{mehta2013online}.
The offline weighted bipartite matching can be solved exactly, 
e.g., via a maximum-flow formulation, while
the best-known approximation ratio for the online variant is achieved by \citet{fahrbach2020edge}.
It is known that many popular variants can be treated as a special case of the 
\emph{submodular social welfare} problem \citep{lehmann2006combinatorial}.
Similar to our reduction in Section~\ref{section:msri}, some assignment or scheduling problems can also be reduced to a subset-selection problem subjective to a bipartite-matching 
constraint~\citep{vondrak2008optimal,pinedo2012scheduling}.

\para{Constrained submodular maximization.}
In the offline setting, for a \nm-matroid constraint, it is well-known that a greedy algorithm guarantees a \nm-approximation for a modular function and a $(\nm+1)$-approximation for a submodular function \cite{korte1978analysis,fisher1978analysis}.
\citet{feldman2011improved} achieve the best-known $(\nm+\epsilon)$-approximation for submodular maximization under a \nm-exchange constraint, which is more general then a \nm-matroid.
\note{Is the $(\nm+\epsilon)$-approximation for submodular function?
How come the approximation for \nm-exchange, which is a more general case, 
is better than the approximation for \nm-matroid?

Guangyi: it is indeed better, with a local-search algorithm. 
That is why it is the best-known result. $\nm+1$-approx by greedy is well-known, but not tight in general.}
A lower bound of $\Omega(\nm/\ln \nm)$ is known for approximating \nm-dimensional matching, which is a special case of maximizing a modular function over a \nm-matroid \citep{hazan2006complexity}.

In regards to the streaming setting,
\citet{levin2021streaming} offer a $3+2\sqrt{2} \approx 5.828$ approximation subject to a matching constraint.
Under a \nm-matroid, a $4\nm$-approximation is obtained by \citet{chakrabarti2015submodular}, which is inspired by the modular variant in \citet{badanidiyuru2011buyback}.
In terms of lower bounds,
the analysis in \citet{badanidiyuru2011buyback} is shown to be optimal for any online algorithm (i.e., a streaming algorithm that maintains only a feasible solution at any moment).
Besides, a 2.692-approximation is impossible even subject to a bipartite matching constraint, and 
there exists evidence that the lower bound can be as high as 3-approximation~\citep{feldman2022submodular}.
For \nm-matroid constraint, a lower bound of $\nm$ has been proven for any streaming algorithm with sub-linear memory;
moreover, any logarithmic improvement over the best-known 4\nm-approximation requires memory super polynomial in \nm \citep{feldman2022submodular}.

%% file: definition.tex
In this section we present the \emph{maximum submodular ranking} (\msr) problem 
in two different streaming models.
Afterwards, we introduce a formulation for progressively-diverse personalized recommendation.
Prior to that, we briefly review notions of submodularity and matroids.

\spara{Submodularity.}
Given a set \V, a function $f: 2^\V \to \reals_+$ is called \emph{submodular} if 
for any $X \subseteq Y \subseteq \V$ and $v \in \V \setminus Y$, it holds $f(v \mid Y) \le f(v \mid X)$,
where $f(v \mid Y) = f(Y+v) - f(Y)$ is the marginal gain of $v$ with respect to set $Y$.
A function $f$ is called \emph{modular} if $f(X) + f(Y) = f(X \cup Y)$ for any $X \cap Y = \emptyset$.
A function $f$ is called \emph{non-decreasing} if 
for any $X \subseteq Y \subseteq \V$, it holds $f(Y) \ge f(X)$.
Without loss of generality, we can assume that function $f$ is normalized, i.e., $f(\emptyset)=0$.

\spara{Matroid.}
For a set \V, a family of subsets $\M \subseteq 2^\V$ is called a \emph{matroid} if it satisfies the following two conditions:
(1)~downward closeness: if $X \subseteq Y$ and $Y \in \M$, then $X \in \M$;
(2)~augmentation: if $X,Y \in \M$ and $|X| < |Y|$, then $X + v \in \M$ for some $v \in Y \setminus X$.
Two useful special cases are those of \emph{uniform matroid} and \emph{partition matroid}.
The former is simply a \nS-cardinality constraint, i.e., 
$\M = \{S \subseteq \V: |S| \le \nS \}$,
and the latter consists of multiple cardinality constraints, 
each placed on a disjoint subset $G_\ell$ of $\V = \cup_\ell G_\ell$, i.e., 
$\M = \{S \subseteq \V: |S \cap G_\ell| \le \nS_\ell, \text{ for all } \ell \}$.
Given \nm matroids $\{\M_j\}_{j \in [\nm]}$, their intersection is called a \emph{\nm-matroid}.

We denote by $\seqs(\V)$ the set of all sequences formed by items in \V.
Given a sequence $\seq \in \seqs(\V)$, 
we write 
$\seq_i$ for the $i$-th item in \seq, and
$\seq + v$ for the new sequence obtained by appending item~$v$ to~\seq.
The length of a sequence \seq is denoted as $|\seq|$.
The set of items in \seq is denoted by $\V(\seq) \subseteq \V$.
A sequence $\seq$ being a subsequence of another sequence $\seq'$ is denoted by $\seq \preceq \seq'$.
Given an interval $w = [s, e]$, where $s, e$ are integers, we write $\seq[w] = \seq[s:e] = \{\seq_s, \ldots, \seq_e\}$.
More generally, given a subset of items $R \subseteq \V$, we write $\seq[R] = \{\seq_i \mid i \in R\}$.


We are now ready to define the \msr problem~\citep{zhang2022ranking} and its streaming variants.
\begin{problem}[Max-submodular ranking (\msr)]
\label{problem:msr}
Given a set \V of \nV items, 
a collection of \nF non-decreasing submodular functions $\fs = \{f_i\}_{i \in [\nF]}$, 
each associated with an integer $\nS_i$, 
the objective is to find a sequence solving
\begin{equation}
	\arg\max_{\seq \in \seqs(\V)} \sum_{f_i \in \fs} f_i(\seq[1 : \nS_i]).
	\label{eq:obj}
\end{equation}
\end{problem}

If an additional \nm-matroid constraint $\M \subseteq 2^\V$ is imposed on items in a feasible sequence $\seq$, i.e.,
$\V(\seq) \in \M$,
we refer to the problem as $\msr\nm$.
Such a \nm-matroid constraint is useful, for example, 
to avoid overrepresentation of some group of items 
in the returned sequence.
If the functions in \fs are modular, we refer to the problem as \emph{maximum modular ranking} (\mmr).

In the function-arriving streaming model, 
we observe a set of new functions $\fs_t$ at time step $t$, 
and the objective is to produce a sequence \seq in real time, that is, 
to decide irrevocably one item in \seq at each time step.
Note that items that are placed in previous item steps cannot be used anymore.
We assume that we observe the new functions $\fs_t$ before deciding the $t$-th item at step $t$, and
a function will stay active in subsequent steps after its arrival until it exhausts its budget.
It is also possible not to place any item at a step (by introducing dummy items in \V).
More formally, the \msr problem in the function-arriving model is defined as follows.

\note{
Just looking problem \ref{problem:msrf}	more carefully. 
What is the streaming/online nature of this problem exactly?
Isn't just solving MSR after each arrival?
We are not penalizing for past decisions since at time $t$
the objective is $f(\seq[t : \nS(f)])$, 
and all the functions that contribute to the objective are available at that point.

Guangyi: 
We actually penalize for past decisions by not allowing to use previously selected items.
This leads to a strong in-approx result.
When we allow an item to be re-used multiple times, it indeed can be interpreted like what you said, solving MSR after each arrival.
}

\begin{problem}[Function-arriving \msr (\msrf)]
\label{problem:msrf}
Given a set \V of \nV items, and 
a collection of non-decreasing submodular functions $\fs_t$ that arrive at the beginning of step $t$, 
with arrival time $\tim(f)=t$ and integers $\nS(f)$ for each $f \in \fs_t$, 
the objective is to find a sequence solving
\begin{equation}
	\arg\max_{\seq \in \seqs(\V)} 
	\sum_t \sum_{f \in \fs_{t}} f(\seq[t : \nS(f)]),
\end{equation}
by irrevocably deciding the $t$-th item $\seq_t$ at step~$t$.
\end{problem}

In contrast, in the item-arriving streaming model, 
we have full information about the functions that are used.
Actually, we further allow each function $f_i$ to ``reserve'' arbitrary $\nS_i$ slots $\slots_i \subseteq [\nV]$ in a sequence, 
instead of merely the first $\nS_i$ slots $[\nS_i]$ in \msr.
When given a sequence, function $f_i$ receives items only from slots $\slots_i$.
For example, when deciding showtimes in a cinema, a user ($f_i$) may only be available during weekends or at specific time of a day.
The goal of the item-arriving \msr problem is to produce a sequence \seq after processing all arriving items in one pass 
and using ``small'' memory size. 
In other words, items that are discarded from the memory cannot be used later.
If there is a slot in the sequence where no function is available, one is allowed to not place any item.
Formally, the \msr problem in the item-arriving streaming model is defined as follows.

\begin{problem}[Item-arriving \msr (\msri)]
\label{problem:msri}
Given a collection of non-decreasing submodular functions $\fs = \{f_i\}_{i \in [\nF]}$, each associated with $\nS_i$ available slots specified by $\slots_i \subseteq [\nV]$, and
items in \V that arrive in a stream,
the objective is to find a sequence in
\begin{equation}
	\arg\max_{\seq \in \seqs(\V)} \sum_{f_i \in \fs} f_i(\seq[\slots_i]).
	\label{eq:obj-msri}
\end{equation}
In addition, the number of items one can store at any moment depends only on $\{ \nS_i \}$ instead of \nV.
\end{problem}

In the offline setting where all items are in place, we call this variant
\emph{\msr with availability} (\msra) problem.
Note that when $\slots_i = [\nS_i] = \{1,\ldots,\nS_i\}$, 
the \msra problem becomes equivalent to the original \msr problem.
We also note that when there is a single function in \fs, 
\msri generalizes the problem of streaming submodular maximization 
for which no streaming algorithm with sublinear memory in \nV 
has an approximation ratio better than 2~\citep{feldman2020one}.

\spara{Progressively-diverse personalized recommendation.}
A popular approach to personalized recommendation \citep{mitrovic2017streaming} 
is to select a succinct subset $S$ of items that maximizes a submodular function of the form 
\begin{align}
f_{\Ctrade,\nS}(S) &= (1-\Ctrade) \sum_{v \in S} \text{rel}(v) + \frac{\Ctrade \nS}{|V|} \, \sum_{u \in \V} \max_{v \in S} \text{sim}(u,v), \nonumber\\
&~\text{ such that }~ |S| \le \nS.
\label{eq:recommend}
\end{align}
Here $\text{rel}(v)$ measures the relevance of an item $v$ to the target user, and
$\text{sim}(u,v)$ the similarity between two items $u$, $v$.
The second term represents one specific notion of diversity 
(also known as representative\-ness or global coverage), 
i.e., for every non-selected item $u \in \V$, 
there exists some item $v \in S$ that is similar enough to $u$.
To create a recommended list that adaptively serves a user's need for diversity, 
one could maximize a weighted sum of multiple functions $\{ f_{\Ctrade,\nS} \}$ 
with increasing trade-off value $\Ctrade \in [0,1]$ and cardinality $\nS$, 
so that the later suffix of the list will be dominated by functions with larger \Ctrade.

%% file: function-arriving.tex
In this section we discuss two scenarios for the function-arriving \msr (\msrf) problem, 
depending on whether items in \V can be used at most once, or more than one time.
We show that the problem is inapproximable in the former case, and 
we present a 2-approximation algorithm for the latter case.

To start off our analysis, 
if the output sequence is constrained to not contain duplicate items, 
it can be shown that the \msrf problem is inapproximable, even when all functions are modular.
This result, stated below, follows from the inapproximability of the \emph{online-selection problem}, 
which aims to select the maximum of an adversarial sequence with no recall \citep{kesselheim2016secretary}. 

\begin{theorem}
\label{theorem:inapprox}
If items can be used at most once,
the \mmrf problem generalizes the online selection problem.
Thus, no randomized algorithm guarantees an $\smallo(\nV)$-approximation for the \mmrf problem.
\end{theorem}

Given the inapproximability of \mmrf for the case that the output 
sequence should not contain duplicates, 
we proceed to study the problem when items in \V can be used multiple times.
This assumption is reasonable in many application, 
for example, a song can be added many times in a playlist --- 
and notice here that unnecessary duplicates are discouraged implicitly as their marginal gain is zero 
with respect to functions that have included these items already. 

\begin{algorithm2e}[t]
\DontPrintSemicolon
Initialize an empty sequence \seq\;
$\fs \gets \emptyset$\;
\For{$t = 1,\ldots$}{
	$\fs \gets \fs \cup \fs_t$ \tcp*{receive functions $\fs_t$}
	$A \gets \{ f \in \fs : \nS(f) \geq t \}$
	\tcp*{active functions at the $t$-th step}
	$v^* \gets \arg\max_{v \in \V} \sum_{f \in A} f(v \mid \seq[\tim(f) : t - 1])$\;
	$\seq \gets \seq + v^*$\;
}
\Return{\seq}\;
\caption{Greedy algorithm for Function-arriving \msr (\msrf)}
\label{alg:msrf-greedy}
\end{algorithm2e}

When duplicates are allowed in the output sequence, 
we prove that a simple greedy algorithm returns a solution 
with a 2-approximation guarantee. 
The greedy, which is displayed as Algorithm~\ref{alg:msrf-greedy},
selects the most beneficial item with respect to the current set of ``active'' functions 
at each step.
A function $f$ is called \emph{active} if it has not exhausted its item budget $\nS(f)$ 
up to that point.

\begin{theorem}
\label{theorem:online-functions}
If items in \V can be used multiple times in the output sequence,
Algorithm~\ref{alg:msrf-greedy} yields a 2-approximation solution for the \msrf problem.
\end{theorem}

The approximation guarantee of the greedy can be shown to be tight, 
as the lower bound provided by \citet{zhang2022ranking} applies to our case, as well. 
In particular, 
since the \msrf problem generalizes the \msr problem, 
by letting all functions to arrive at the beginning, 
we obtain the following result.
\begin{remark}[\citet{zhang2022ranking}]
\label{remark:tight}
If an item can be used multiple times in the output sequence,
the 2-approximation solution obtained by Algorithm~\ref{alg:msrf-greedy} is tight for the \msrf problem.
\end{remark}

In the rest of this section, we prove Theorem~\ref{theorem:online-functions},
and the proof of Theorem~\ref{theorem:inapprox} is deferred to Section~\ref{section:theorem:inapprox}.

\begin{proof}[Proof of Theorem \ref{theorem:online-functions}]
We write $A_t=\{ f \in \bigcup_{t' \le t} \fs_{t'}: \nS(f) \geq t \}$ 
for the set of active functions at step $t$.
We denote by \seq the sequence produced by Algorithm~\ref{alg:msrf-greedy} with objective value \ALG, and
by $\seq^*$ the optimal sequence with objective value \OPT.

By the greedy selection criterion, 
we know that for any arbitrary item $v \in \V$, it holds that
\begin{equation}
    \sum_{f \in A_{t}} f(\seq_{t} \mid \seq[\tim(f): t-1] )
    \ge \sum_{f \in A_{t}} f(v \mid \seq[\tim(f): t-1]).
    \label{eq:msrf-greedy}
\end{equation}

To simplify the notation, let us write $\seq[f]$ to mean $\seq[\tim(f) : \nS(f)]$.

It is easy to see that \ALG is equal to the sum over $t$ of the left-hand side in Equation~(\ref{eq:msrf-greedy}), implying

\note{I do not see what the first inequality is by Equation~(\ref{eq:msrf-greedy}).
Equation~(\ref{eq:msrf-greedy}) does not involve $\seq^*_t$.

Guangyi: item $v$ in Equation~(\ref{eq:msrf-greedy}) can be arbitrary, so we replace it with each corresponding item in $\seq^*_t$.}
\begin{align*}
	\ALG 
	&= \sum_t \sum_{f \in A_{t}} f(\seq_{t} \mid \seq[\tim(f): t-1] ) \\
	&\stackrel{(a)}{\ge} \sum_t \sum_{f \in A_{t}} f(\seq^*_t \mid \seq[\tim(f): t-1] ) \\
	&= \sum_{f \in \fs} \sum_{t=\tim(f)}^{\nS(f)} f(\seq^*_t \mid \seq[\tim(f): t-1]) \\
	&\stackrel{(b)}{\ge} \sum_{f \in \fs} \sum_{t=\tim(f)}^{\nS(f)} f(\seq^*_t \mid \seq[f]) \\
	&\stackrel{(c)}{\ge} \sum_{f \in \fs} f(\seq^*[f] \mid \seq[f]) \\
	&= \sum_{f \in \fs} f(\seq[f] \cup \seq^*[f]) - f(\seq[f]) \\
	&\ge \sum_{f \in \fs}  f(\seq^*[f]) - f(\seq[f]) 
	= \OPT - \ALG,
\end{align*}
where 
inequality (a) is by Eq.~(\ref{eq:msrf-greedy}), and
inequalities (b) and (c) are due to submodularity.
\hfill\end{proof}

%% file: reduction.tex
\begin{table*}[t]
\caption{Summary of approximation ratios for \msra and \msri}
\label{tbl:reduction}
\centering
\begin{tabular}{lccccc}
\toprule
	&unconstrained	&\nm-matroid	\\
\midrule
\mmra	&exact	&$\nm+1$ (\citet{korte1978analysis})	\\
\msra	&$2+\Cgap$ (\citet{feldman2011improved})	&$\nm+1+\Cgap$ (\citet{feldman2011improved})	\\
\mmri	&$1/0.5086$ (\citet{fahrbach2020edge})	&$2(\nm+1+\sqrt{(\nm+1)\nm})-1$ (\citet{badanidiyuru2011buyback})	\\
\msri	&$5.828$ (\citet{levin2021streaming})	&4(\nm+1) (\citet{chakrabarti2015submodular})	\\
\bottomrule
\end{tabular}
\end{table*}

In this section, we present the reduction that turns the ranking problem into a subset selection problem subject to a bipartite matching constraint, and its rich consequences.
A summary of approximation ratios in different settings is displayed in Table \ref{tbl:reduction}, using algorithms provided in the citations.

\begin{theorem}
\label{theorem:MSRp}
For any integer $\nm \ge 1$, the $\msra\nm$ problem is an instance of maximizing a non-decreasing submodular function subject to a $(\nm+1)$-matroid.
\end{theorem}

As an immediate consequence of Theorem \ref{theorem:MSRp},
the item-arriving \msr (\msri) can be solved by streaming algorithms for constrained submodular maximization.
\begin{corollary}
\label{corollary:MSR-V}
For any integer $\nm \ge 1$, the $\msri\nm$ problem is an instance of maximizing a non-decreasing submodular function subject to a $(\nm+1)$-matroid in one pass while remembering $\bigO(\sum_i \nS_i)$ items at any moment.
\end{corollary}

Moreover, when functions are modular, we can obtain a stronger result.
\begin{corollary}
\label{corollary:MMR}
For any integer $\nm \ge 1$, the $\mmra\nm$ problem is an instance of maximizing a modular function subject to a $(\nm+1)$-matroid.
In particular, the $\mmra$ problem is an instance of maximum weighted bipartite matching.
\end{corollary}
\note{Modular maximization over 3-matroid generalizes \np-hard 3-dimensional matching.}


\begin{proof}[Proof of Theorem \ref{theorem:MSRp}]
The main idea of the reduction is to create an extended universe set $\V'$, i.e.,
\begin{equation}
\label{eq:extended-V}
	\V' = \{ (v,t) : v \in \V, \, t \in [\nV] \},
\end{equation}
which can be seen as the edges in a complete bipartite between items $L = \V$ and ranks $R = [\nV]$.
Let us define $\V(S') = \{ v : (v,t) \in S' \}$ to be the projection of $S'$ onto \V.
Let us also write $X_v = \{(v, t) : t \in [\nV]\}$ and $Y_t = \{(v, t) : v \in \V\}$.

Suppose we are given a \nm-matroid $\M \subseteq 2^\V$ over \V.
Define $\M' = \mathcal{A} \cap \mathcal{B} \cap \mathcal{C}$, where
\[
\begin{split}
	\mathcal{A} & = \{ S' \subseteq \V': \V(S') \in \M \}, \\
	\mathcal{B} & = \{ S' \subseteq \V': | S' \cap X_v | \le 1, \text{ for all } v  \in \V  \}, \\
	\mathcal{C} & = \{ S' \subseteq \V': | S' \cap Y_t | \le 1, \text{ for all } t  \in [\nV] \}. \\
\end{split}
\]
Here, $\mathcal{B}$ forces that an item appears only once and $\mathcal{C}$ forces that only one item appears at time $t$.
Consequently, a feasible sequence \seq can be written as a subset of $\V'$ that satisfy $\M'$.

On the other hand, a feasible subset $S' \subseteq \V'$ can be transformed to a sequence by ordering items in $\V(S')$ according to their associated ranks in $S'$.
If $S'$ consists of non-consecutive ranks, one can insert dummy items (or shifting items forward in \msr) to obtain a sequence, with no decrease in the objective function.

We claim that $\M'$ is a $(\nm+1)$-matroid. We will prove the claim
by arguing that $\mathcal{A} \cap \mathcal{B}$ is a $\nm$-matroid 
(by Lemma~\ref{lemma:AB} in Section~\ref{section:theorem:MSRp}). 
Then $\mathcal{A} \cap \mathcal{B} \cap \mathcal{C}$ is a $(\nm+1)$-matroid because $\mathcal{C}$ is a matroid.

What is left is to show that the objective function for \msra (Equation~\ref{eq:obj-msri}) is non-decreasing and submodular with respect to the new universe set $\V'$.
Showing non-decreasing is obvious, so we only elaborate on submodularity.
Recall that $\slots_i \subset [\nV]$ consists of available slots for function $f_i$.
Let us define $\slots'_i = \{ (v,t) : v \in \V, \, t \in \slots_i \}$.
The objective function can be now written as
	$g(S) = \sum_{f_i \in \fs} f_i(\V(S \cap \slots_i'))$.
For any subset $S \subseteq W \subseteq \V'$, 
the marginal gain $g(\cdot \mid S)$ of including an item $(v,t)$ into $S$ is
\begin{align*}
	&g((v,t) \mid S) = \sum_{f_i \in \fs: t \in \slots_i} f_i(v \mid \V(S \cap \slots_i')) \\
	&\ge \sum_{f_i \in \fs: t \in \slots_i} f_i(v \mid \V(W \cap \slots_i')) 
	= g((v,t) \mid W),
\end{align*}
since $\V(S \cap \slots_i') \subseteq \V(W \cap \slots_i')$ and submodularity of each $f_i$.
Hence, the $\msra\nm$ problem can be cast as an instance of non-decreasing submodular maximization under a $(\nm+1)$-matroid.
\hfill\end{proof}

%% file: experiment.tex
We evaluate our methods on novel use cases
that motivate the \msrf and \msri problems.
For each use case,
we simulate a concrete task using real-life data, and 
empirically evaluate the performance of the proposed algorithms.
A summary of the datasets can be found in Table~\ref{tbl:datasets}.
An examination on the running time is deferred to Section~\ref{section:runtime}.
Our implementation has been made publicly available.\code

\begin{table}[t]
  \caption{Datasets statistics}
  \label{tbl:datasets}
  \centering
\begin{tabular}{lrr}
\toprule
Dataset & $\nV = |\V|$ 
        & $\nF = |\fs|$ \\
\midrule
Music \cite{Bertin-Mahieux2011}	&61 415	&10 000	\\
Github social network \cite{rozemberczki2019multiscale}	&37 700	&100	\\
Sogou web pages \cite{liu2011users}	&725	&1 017	\\
Twitter words \cite{pennington2014glove}	&10 000	&8	\\
\bottomrule
\end{tabular}
\end{table}

\begin{figure}[t]
    \centering
    \subcaptionbox{Music live streaming}{\includegraphics[width=.4\textwidth]{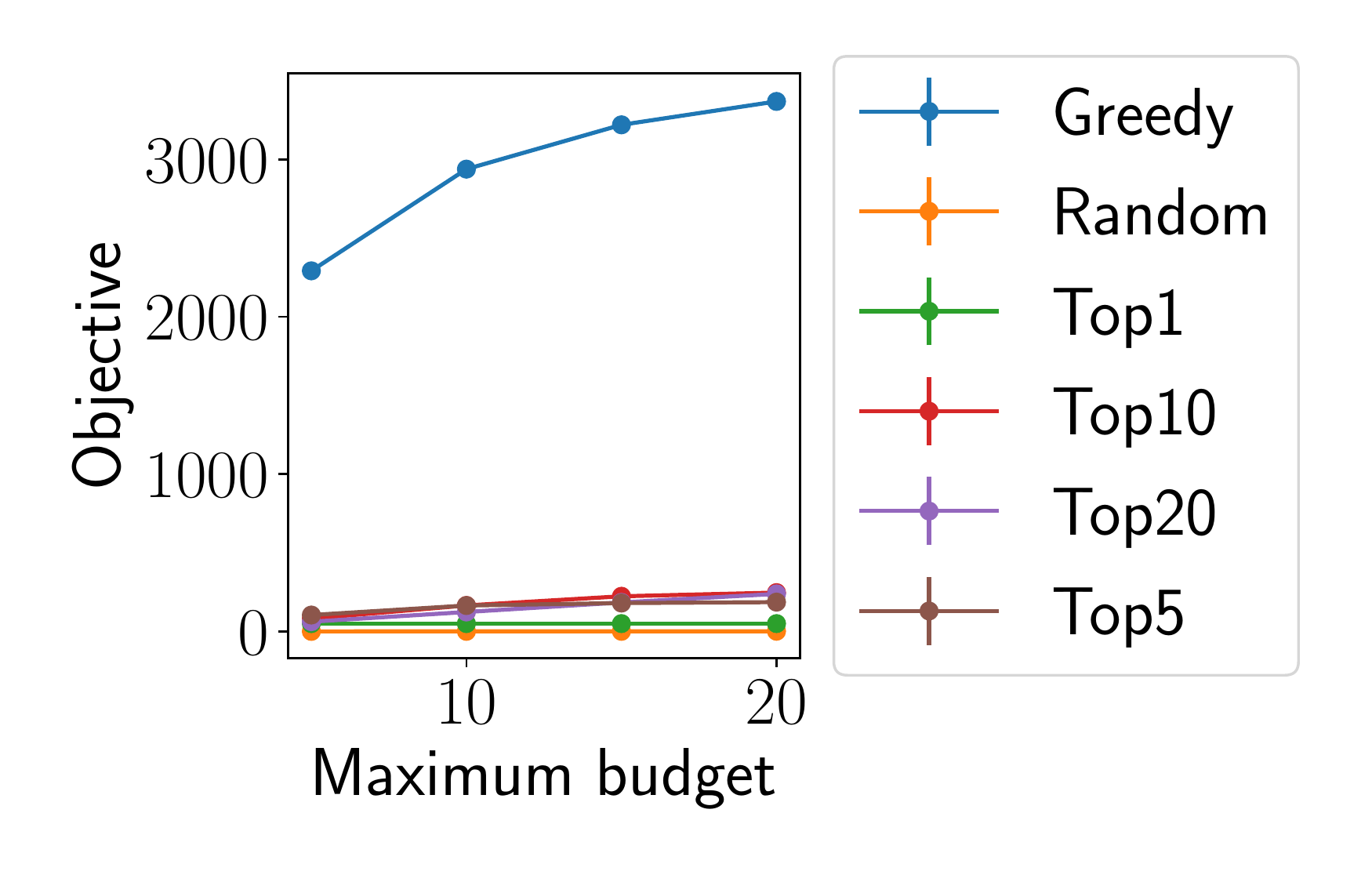}}
    \quad
    \subcaptionbox{Catalogued viral marketing in Github}{\includegraphics[width=.4\textwidth]{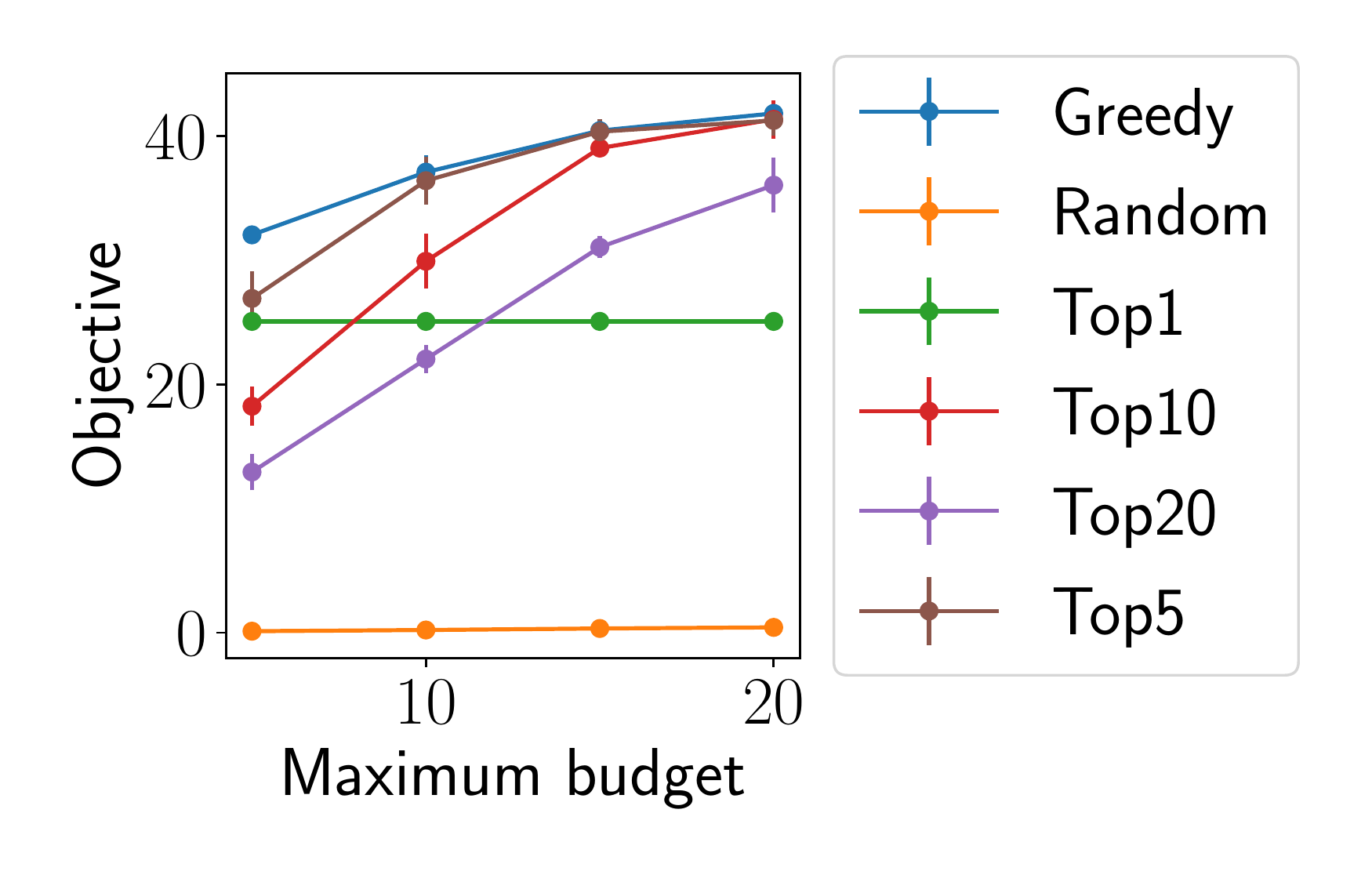}}
    \caption{\msrf. 
    Each demand is given a random arrival time, and a random budget between 1 to a parameter of maximum budget.
    (a) Utility of a user demand is $f(S) = |S \cap S_{\text{like}}| / |S_{\text{like}}|$, where $S_{\text{like}}$ is the set of songs the user likes.
    (b) Utility of a marketing demand is $f(S) = |N(S) \cap S_{\text{g}}| / |S_{\text{g}}|$, 
    where $N(S)$ is neighborhood of $S$ and $S_{\text{g}}$ is the target group of nodes in the network.
    }
    \label{fig:msrf}
\end{figure}

\subsection{Function-arriving \msr}
We present two use cases for the \msrf problem, 
live streaming and catalogued viral marketing, 
and we evaluate our methods on relevant datasets.

\para{Algorithms.}
The proposed greedy algorithm in Algorithm~\ref{alg:msrf-greedy} is termed \greedy.
Other baselines include 
\random, which picks a random item at every step, and
\looptopk, which selects the top-$k$ items repetitively --- 
in our use case this corresponds to playing the most popular songs in a loop.
Note that \looptopk is \emph{omniscient} as it requires information about the item popularity in advance.

\para{Live streaming.}
One increasingly popular application on the internet is live streaming, 
where a live streamer performs various shows continuously, 
while the audience may join or leave any time.
More concretely, we consider music live streaming, 
where the live streamer plays songs continuously.
To simulate this application, we use the Million Song Dataset \cite{Bertin-Mahieux2011}, 
consisting of triples representing a \emph{user}, \emph{song}, and \emph{play count}.
We assume that a user \emph{likes} a song if it is played more than once.
We define \emph{user utility} to be fractional coverage of the liked songs, 
which is a submodular function.
We set a random budget for each user between 1 to a maximum-budget parameter, 
i.e., how many songs the user will listen to.
Besides, every user is given a random arrival time over a long horizon.
Our goal is decide a sequence of songs in real time that maximizes total user utility.

\para{Catalogued viral marketing.}
For viral-marketing applications in social networks,
the goal is to identify a small set of seed nodes who can influence many other users.
For popular diffusion models, the number of influenced nodes is a submodular function
of the seed set \cite{kempe2015maximizing}.
Here, 
we introduce a generalization of this classic result
to cope with multiple marketing demands simultaneously,
where each demand is only interested in reaching a specific group of users.
For example, an advertiser may want to influence female users, 
while another may want to influence users near a specific city.
Each demand provides a number of product samples to seed nodes, 
with the hope to advertise their products by word of mouth.

As a use case for our experimental evaluation, 
we consider a platform designed to help with these marketing demands, and 
assume that a package of samples of different products
can be sent to one seed node at each time step.
As the marketing demands arrive in real time, 
we aim to catalogue unfinished demands by identifying a seed node 
that is beneficial to all of them.
Note that samples from one demand should be distributed as soon as possible after its arrival
to the outgoing packages.
To simulate this use case, we use the GitHub social network~\cite{rozemberczki2019multiscale}.
We consider 100 demands, each targeting a random subset of the network, 
together with a random number of samples from 1 to a maximum-budget parameter, 
and a random arrival time.
Our goal is to maximize the total utility of demands by sending a package 
to one carefully chosen seed node at each time step.

\para{Results.}
The result of the simulation is shown in Figure~\ref{fig:msrf},
in which each data point represents an average of three runs, 
each with a different random seed.
For the music-streaming task (Figure~\ref{fig:msrf}(a)), 
the \greedy algorithm outperforms all other baselines by a large margin.
This suggests that for users with diverse preferences in songs, 
an algorithm like \greedy, which can adapt to the need of current active users, 
is required for good performance.
In the catalogued viral-marketing task (Figure~\ref{fig:msrf}(b)), 
the \greedy algorithm continues to achieve the best performance.
However, several baselines from \looptopk come closer to \greedy as the demand budget increases.
This signifies the existence of a group of influencers 
who can collectively reach the most users in the Github network.

\begin{figure}[t]
    \centering
    \subcaptionbox{Web page ranking in Sogou}{\includegraphics[width=.4\textwidth]{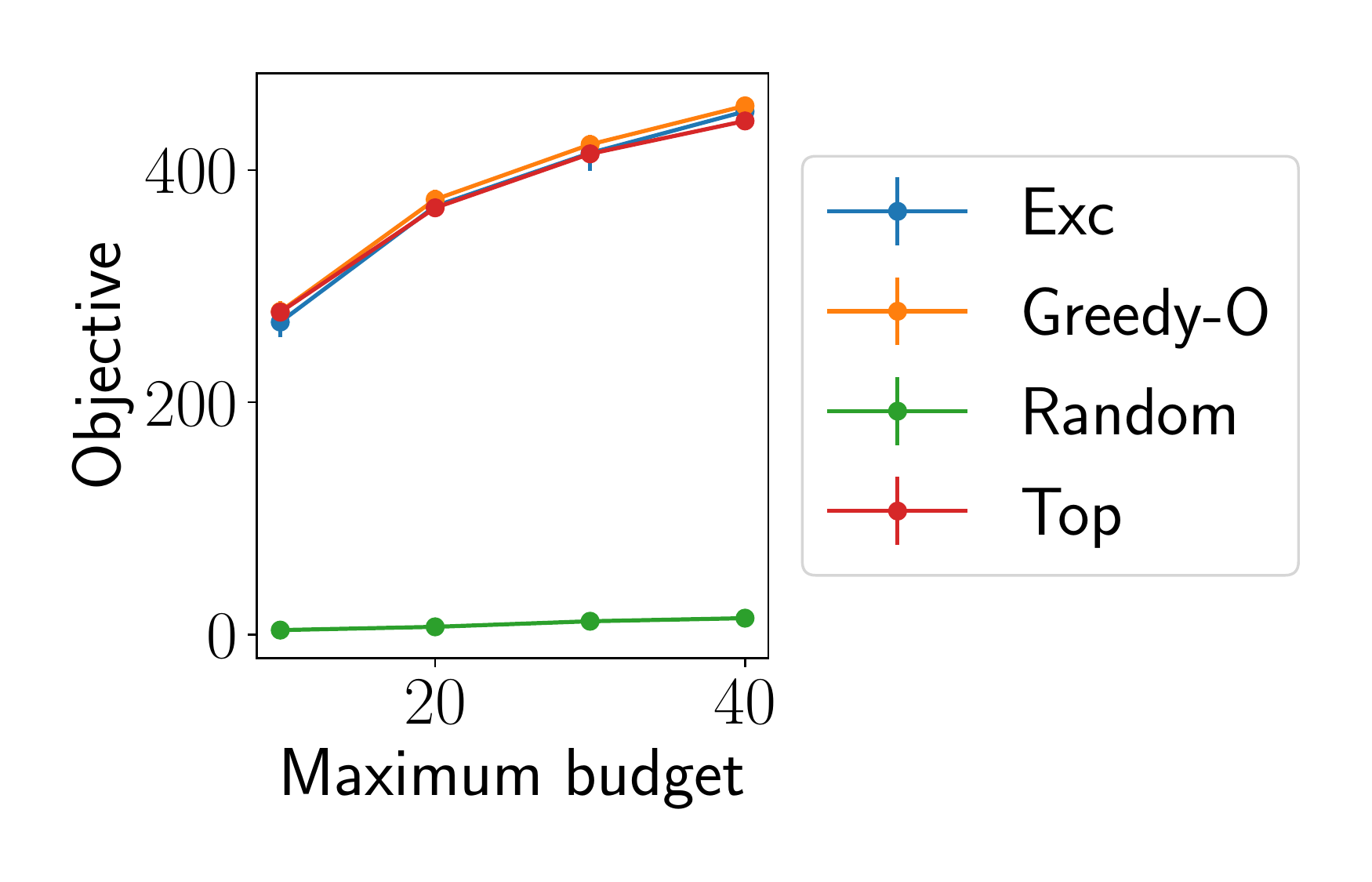}}
    \subcaptionbox{Finding synonyms to ``Trump'' in Twitter}{\includegraphics[width=.5\textwidth, trim = 1cm 0cm 0.9cm 0cm, clip]{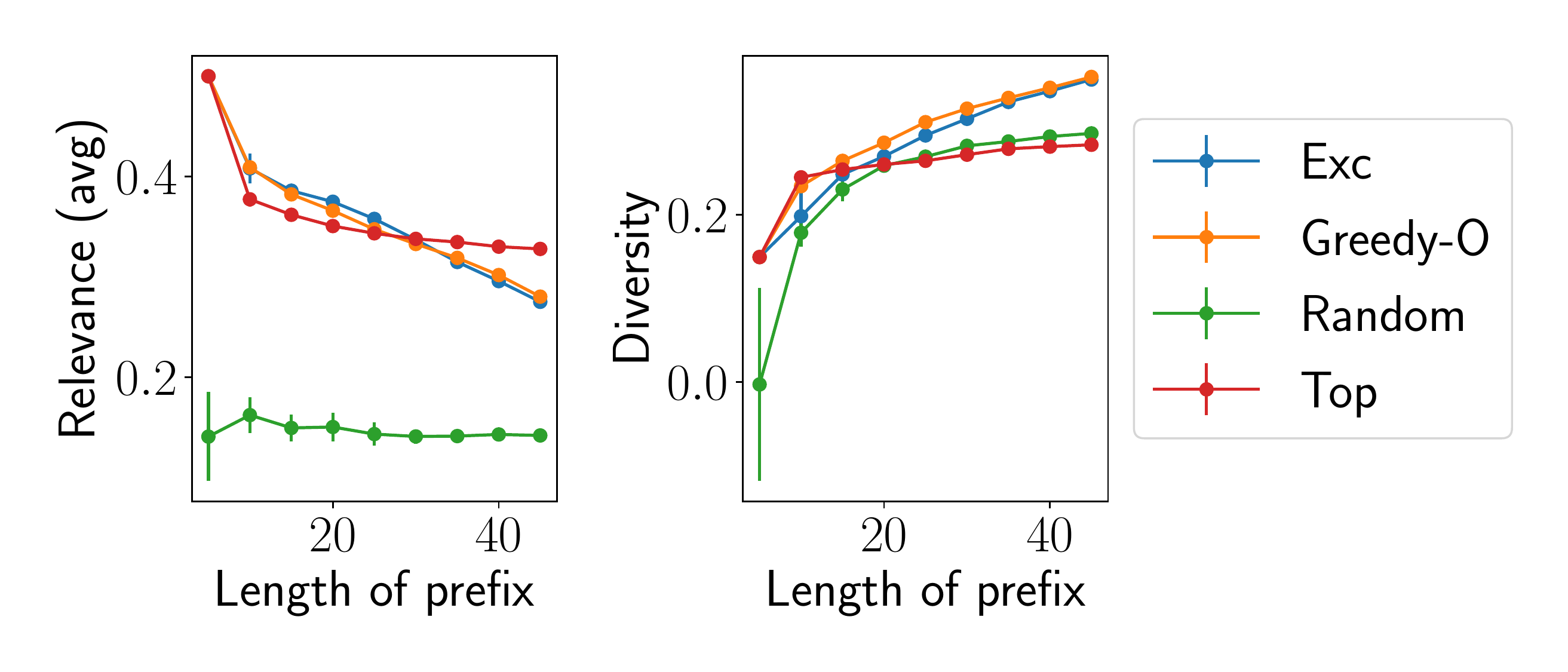}}
    \caption{\msri.
    (a) Utility of a user intent is $f(S) = |S \cap S_{\text{rel}}| / |S_{\text{rel}}|$, where $S_{\text{rel}}$ is the set of relevant pages.
    Each intent is given a random budget between 1 to a parameter of maximum budget.
    (b) Given a sequence, its relevance and diversity (Equation~\ref{eq:recommend}) is measured at every prefix.
    }\label{fig:msri}
\end{figure}

\subsection{Item-arriving \msri.}
Next we will present experiments for \msri as well as for progressively-diverse personalized recommendation.

\para{Algorithms.}
We adopt the state-of-the-art streaming algorithm in \citet{chakrabarti2015submodular} as our algorithm,
which greedily assigns each arriving item to one of ranks in the sequence whenever possible, 
starting from rank 1.
We call this algorithm Exchange (\exc).
If a rank is occupied by some previous item, 
\exc replaces it if the current item is twice more valuable than the existing item.
Other baselines include 
\random, which produces a random sequence, and
\topk, which orders items by non-increasing singleton utility.
We also include an offline baseline, Omniscient Greedy (\greedyo) \cite{zhang2022ranking}, which serves as a tighter estimate for the optimal value.
\greedyo sequentially selects greedily an item with respect to the current set of active functions.

\para{Multiple intents re-ranking.}
In the absence of explicit user intent for a given query,
a search engine needs to take into account all possible intents when providing a list of returned web pages.
Each intent is only relevant to a subset of web pages.
We represent the utility of an intent by the fractional coverage 
of relevant pages browsed before running out of patience.
The goal is to produce a list of web pages to maximize the utility over all intents.

In this use case, we extract user intents from the SogouQ click log dataset \cite{liu2011users}.
We consider all queries related to ``movie,'' and 
for each such query
we collect all users who issued the query.
We also collect the pages they clicked.
We treat each user as an intent, and pages they clicked as the set of relevant pages.
For each intent, we generate a random number from 1 to a maximum-budget parameter as 
user ``patience,'' i.e., the number of pages that the user will browse. 

\para{Progressively-diverse personalized recommendation.}
Next we consider the recommendation task introduced in Section~\ref{section:definition}.
To simulate a concrete task, we consider the task of finding representative synonyms with respect to a given keyword.
We choose ``Trump'' as our keyword.
We use the pre-trained word embedding Glove over the Twitter corpus \cite{pennington2014glove}.
Similarity or relevance between any two words is measured by the cosine similarity minus 0.5 due to dense vectors.
The top 10 000 relevant words form our candidate set \V, among which 100 random words are chosen as bases to compute the diversity term (Equation~\ref{eq:recommend}).
Given a maximum length ($\nS=40$) of the recommended list,
we create multiple functions $\{ f_i \}$ (Equation~\ref{eq:recommend}) for $i \le 8$,
where the $i$-th function is associated with a weight $(1/2)^i$, a trade-off value $\Ctrade=(i-1)/\nS$ and a budget of $5i$.
Thus, function~$f_i$ is dominant for the $i$-th length-5 subsequence, and $f_i$ with a large $i$ favors increasingly diverse items.

\para{Results.}
The results are shown in Figure~\ref{fig:msri}.
Every data point represents an average of three runs, each with a different random seed.
For the task of web page ranking (Figure~\ref{fig:msri}(a)), all algorithms except for \random perform almost equally well.
This implies a heavy overlap in relevant pages among different user intents.
For the task of finding diverse synonyms (Figure~\ref{fig:msri}(b)), 
the ranking in terms of the objective is 
$\greedyo \approx \exc > \topk > \random$ 
($7.139 \approx 7.136 > 6.652 > 3.527$).
We demonstrate two components of the objective, relevance and diversity, separately in Figure~\ref{fig:msri}(b).
Note that the \random algorithm is a classic method in finding diverse representatives, while
the \topk algorithm is optimal if the objective degenerates into a single modular term of relevance.
The word list returned by the \exc algorithm is indeed increasingly the most diverse, while it also finds relevant synonyms at the beginning.
The actual word list returned by \exc is presented in Section~\ref{section:synonyms}.

%% file: conclusion.tex
In this paper, we study extensions of the \msr problem in two streaming models.
In the first demand-arriving model,
we show that a greedy algorithm guarantees 2-approximation, if items can be reused.
In the second item-arriving model,
we discover a reduction that turns the ranking problem into a constrained subset-selection problem, and
inherit approximation guarantees from standard submodular maximization.
The reduction further allows us to approximate the \msr problem subjective to additional \nm-matroid constraints.
Finally, we describe several novel applications for the \msr problem, and examine empirical performance of the proposed algorithms.

One limitation of the \msr formulation is that individual budgets are not always known in some applications.
Another limitation is that it may be computationally costly if both the number of demands and items are large, especially for the item-arriving \msr problem.

With respect to ethical considerations of the work, our algorithm is a general submodularity-based framework for ranking,
which is not dedicated to a specific application. 
We cannot identify strong negative societal concerns. 
Many of the broader machine-learning issues, such as misuse of technology, biases in data, effects of automation in the society, and so on, are relevant to this work, as well, but in no greater degree than the whole machine-learning field.

%% file: supp.tex
\subsection{Proof of Theorem \ref{theorem:inapprox}}
\label{section:theorem:inapprox}
\begin{proof}[Proof of Theorem \ref{theorem:inapprox}]
For the online-selection problem,
the input is a sequence of $\nV$ integer numbers $a_1,\ldots,a_\nV$
that are revealed one after another. 
The problem asks to select exactly one number immediately after it is revealed, 
and the goal is to maximize the value of the selected number.
It is well-known that this problem does not admit a $\smallo(\nV)$ approximation \citep{kesselheim2016secretary}.
\note{For a proof, see \url{https://www.mpi-inf.mpg.de/fileadmin/inf/d1/teaching/summer16/random/yaosprinciple.pdf}.}

We observe that the online-selection problem is a special case of \mmrf.
To see this, consider the following instance:
let the universe set \V consisting of an item $v$ and $\nV-1$ other dummy items.
Every arriving modular function $f$ has an identical form, with $f(v)=a$ and $f(u)=0$ for any item $u \ne v$.
Besides, every function $f$ is associated with a budget of $\nS(f)=1$.
At the $t$-th step, a function as defined above with $f(v)=a_t$ arrives.
It is easy to see that the optimal objective value of \mmrf is equal to the largest number $a_t$.
Then, deciding the rank of item $v$ in the output sequence for the \mmrf problem
is equivalent to deciding which number in $\{a_t\}$ to select for the online-selection problem.
Hence, \mmrf generalizes the online-selection problem.
\hfill\end{proof}
\note{Number $a_t$ could go up to $\bigO(C^\nV)$ for constant $C$.}

\subsection{Missing proof in Theorem \ref{theorem:MSRp}}
\label{section:theorem:MSRp}

\begin{lemma}
\label{lemma:AB}
$\mathcal{A} \cap \mathcal{B}$ is a $\nm$-matroid.
\end{lemma}
\begin{proof}
Since $\M$ is a $\nm$-matroid over \V, then we can write
$\M = \M_1 \cap \cdots \cap \M_{\nm}$ where $\M_i$ a matroid.
For each matroid $\M_i$, we construct another system over $\V'$,
\begin{align*}
	&\M_i' = \{ S' \subseteq \V': \V(S') \in \M_i \} \cap \mathcal{B} \\
	&= \{ S' \subseteq \V': \V(S') \in \M_i, | S' \cap X_v | \le 1, \text{ for all } v  \in \V \}.
\end{align*}
Notice that $\mathcal{A} = \bigcap_i \{ S' \subseteq \V': \V(S') \in \M_i \}$, and 
$\mathcal{A} \cap \mathcal{B} = \bigcap_i \M_i'$.
We prove that $\mathcal{A} \cap \mathcal{B}$ is a \nm-matroid by verifying that each $\M_i'$ is a matroid.

To show that $\M_i'$ is a matroid,
downward closeness is obvious, and we only verify augmentation.
Given any $T, U \in \M_i'$ such that $|T| < |U|$,
we have that $|\V(T)| = |T| \le |\V(U)| = |U|$.
Since $\V(T), \V(U) \in \M_i$, there exists $v \in \V(U)$ such that $v + \V(T) \in \M_i$.
Suppose $(v,t) \in U$ for that particular $v$, and it is obvious that $(v,t) + T \in \M_i'$.
Hence, $\M_i'$ is a matroid, implying that $\mathcal{A} \cap \mathcal{B} = \bigcap_i \M_i'$ is a \nm-matroid.
\hfill\end{proof}

\subsection{Synonyms to ``Trump'' in Twitter}
\label{section:synonyms}

trump
banks
warren
clinton
gates

newman
buffett
founder
reagan
carson

ceo
appoints
butcher
duffy
carlson

lowe
travis
costello
joins
airbnb

company
tesla
sanford
krause
dunlap

cassidy
does
shipbuilding
shooter
hired

rwanda
asml
hartman
barb
grandfather

rig
exchanging
lowes
varela
lamontagne

\subsection{Running time}
\label{section:runtime}

\begin{figure}[h]
    \centering
	\includegraphics[width=.5\textwidth]{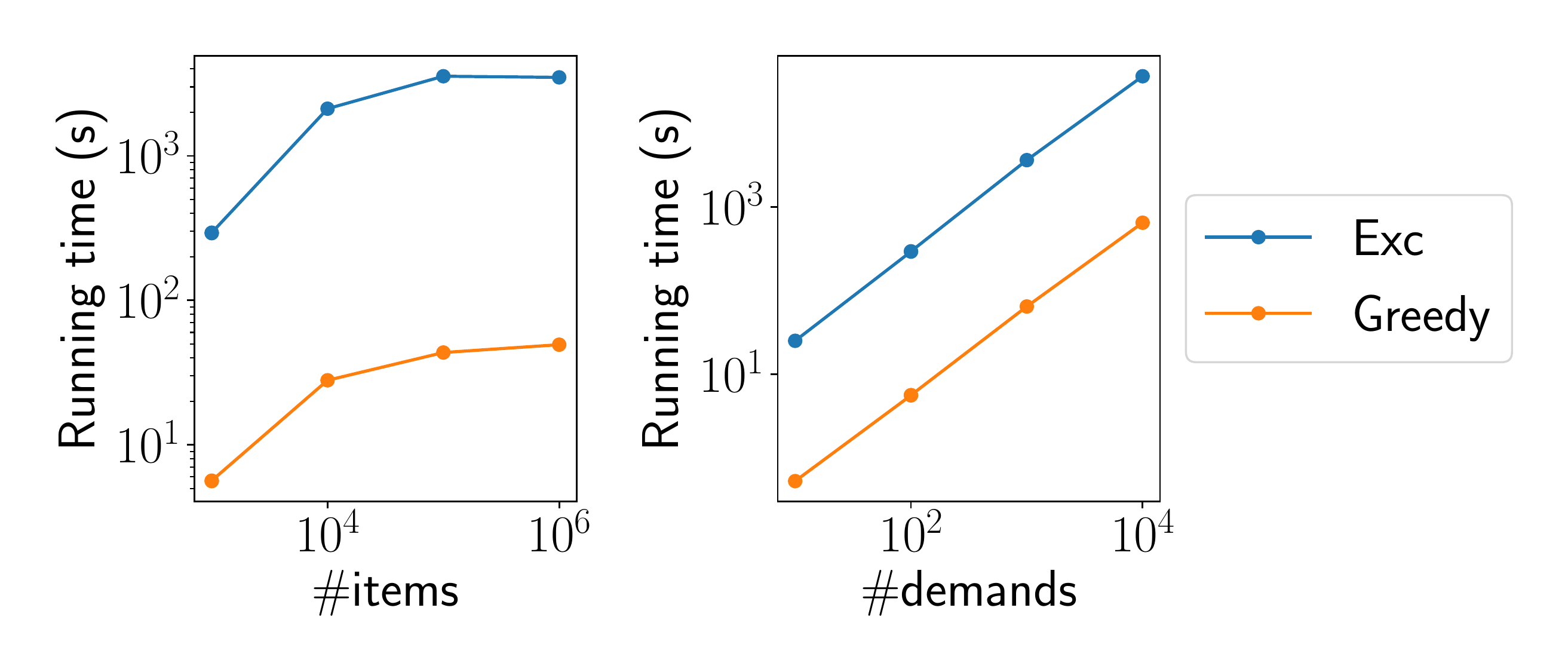}
    \caption{Running time}
    \label{fig:runtime}
\end{figure}

To examine the running time of the proposed algorithms, we generate synthetic data with an increasing number of items or demands.
More specific, we either fix the number of demands (100) while increase the number of items, 
or fix the number of items (1000) while increase the number of demands.
Each demand is represented by a coverage function of a random subset of size 100, and is assigned a random budget between 1 to 100.
For the \msrf model (\exc), a random arrival time is given to each demand.

The results are displayed in Figure~\ref{fig:runtime}.
Running time of both \exc and \greedy algorithms increases linearly in the number of demands.
As to an increasing number of items, running time of both appears to be sub-linear instead of linear, which is due to the fact that many items fail to hit any demand subsets.

\subsection{Additional experimental details}
\label{section:experiment-details}

All experiments were carried out on a server equipped with 24 processors of AMD Opteron(tm)
Processor 6172 (2.1 GHz), 62GB RAM, running Linux~2.6.32-754.35.1.el6.x86\_64.
We use Python~3.8.5.